\documentclass[a4paper, 12pt]{article}

\def\margin{2cm}
\usepackage[left=\margin,right=\margin,top=\margin,bottom=\margin]{geometry}

\usepackage{authblk}
\usepackage{amsmath, amssymb, amsthm}
\usepackage{complexity}
\allowdisplaybreaks

\usepackage{pstricks}
\psset{gridcolor=yellow, gridlabelcolor=yellow, subgridcolor=yellow,% 
       gridlabels=5pt, subgriddots=3,gridwidth=0.5pt}
\usepackage{graphicx}

\usepackage{enumerate}
\usepackage{arrayjobx}

\title{Rainbow Colouring of Split Graphs}
\date{}
\author[1]{L.~Sunil~Chandran}
\author[2]{Deepak~Rajendraprasad \footnote{Supported by VATAT Post-doctoral Fellowship, Council of Higher Education, Israel.}}
\author[3]{Marek Tesa\v{r}\footnote{Supported by Charles University by the grant SVV-2014–260103}}
\affil[1]{
	Department of Computer Science and Automation, \authorcr 
	Indian Institute of Science, Bangalore, India - 560012. \authorcr
	\texttt{sunil@csa.iisc.ernet.in}
}
\affil[2]{
	The Caesarea Rothschild Institute, Department of Computer Science, \authorcr
	University of Haifa, 31095, Haifa, Israel. \authorcr
	\texttt{deepakmail@gmail.com}
}
\affil[3]{
	Department of Applied Mathematics, Faculty of Mathematics and Physics, \authorcr
	Charles University, Prague, Czech Republic. \authorcr
	\texttt{tesar@kam.mff.cuni.cz}
}

\theoremstyle{plain}
\newtheorem{theorem}{Theorem}
\newtheorem{lemma}[theorem]{Lemma}
\newtheorem{corollary}[theorem]{Corollary}

\theoremstyle{definition}
\newtheorem{definition}[theorem]{Definition}
\newtheorem{problem}{Problem}

\newtheoremstyle{cases}% name of the style to be used
  {}% measure of space to leave above the theorem. E.g.: 3pt
  {}% measure of space to leave below the theorem. E.g.: 3pt
  {}% name of font to use in the body of the theorem
  {}% measure of space to indent
  {}% name of head font
  {\newline}% punctuation between head and body
  {0.5em}% space after theorem head
  {{\itshape \thmname{#1}} \thmnumber{#2} ({\itshape\thmnote{#3}}).\medskip}% Manually specify head
\theoremstyle{cases}
\newtheorem{case}{Case}
\newtheorem*{claim}{Claim}

\def\into{\rightarrow}

\def\-{\mbox{--}}
\def\G{\mathcal{G}}
\def\S{\mathcal{S}}
\def\R{\mathbb{R}}
\def\Z{\mathbb{Z}}

\DeclareMathOperator{\rc}{\textsc{RainbowColour}}
\DeclareMathOperator{\bcc}{\textsc{BicliqueCover}}
\DeclareMathOperator{\sat}{3\textsc{-Sat}}

\DeclareMathOperator{\pen}{pen}
\DeclareMathOperator{\then}{then}
\newcommand{\Edge}[1]{\stackrel{#1}{\mbox{---\negthinspace---}}}
\newcommand{\Path}[2]{#1 \textnormal{ to } #2}
\newcommand{\PrintColours}[1]{{\small \color{red} #1}}

\def\Obr#1#2#3#4%
{%
\begin{figure}[htp]
\centering
\includegraphics[scale = #1]{./pictures/#2}
\caption{#3}
\label{#4}
\end{figure}
}%

\def\GOneOneOne{
\begin{pspicture}(0,0)(4,3)
	\pspolygon(2,2)(1.4,1)(2.6,1)
	\psline[showpoints=true,dotsize=5pt](2,2)(2,3)
	\psline[showpoints=true,dotsize=5pt](1.4,1)(0.8,0)
	\psline[showpoints=true,dotsize=5pt](2.6,1)(3.2,0)
	\uput[r](2,2){$x_0$}
	\uput[r](2.6,1){$x_1$}
	\uput[l](1.4,1){$x_2$}
	\uput[u](2,3){$y_0$}
	\uput[dr](3.2,0){$y_1$}
	\uput[dl](0.8,0){$y_2$}
	\uput[u](2,-1){$G_{111}$}
\end{pspicture}
}

\def\GFourZeroZero{
\begin{pspicture}(0,0)(4,3)
	\pspolygon[showpoints=true,dotsize=5pt](2,2)(1.4,1)(2.6,1)
	\psline[showpoints=true,dotsize=5pt](2,2)(1.1,3)
	\psline[showpoints=true,dotsize=5pt](2,2)(1.7,3)
	\psline[showpoints=true,dotsize=5pt](2,2)(2.3,3)
	\psline[showpoints=true,dotsize=5pt](2,2)(2.9,3)
	\uput[r](2,2){$x_0$}
	\uput[r](2.6,1){$x_1$}
	\uput[l](1.4,1){$x_2$}
	\uput[u](1.1,3){$y_0$}
	\uput[u](1.7,3){$y_1$}
	\uput[u](2.3,3){$y_2$}
	\uput[u](2.9,3){$y_3$}
	\uput[u](2,-1){$G_{400}$}
\end{pspicture}
}

\def\GThreeOneZero{
\begin{pspicture}(0,0)(4,3)
	\pspolygon[showpoints=true,dotsize=5pt](2,2)(1.4,1)(2.6,1)
	\psline[showpoints=true,dotsize=5pt](2,2)(2,3)
	\psline[showpoints=true,dotsize=5pt](2,2)(1.4,3)
	\psline[showpoints=true,dotsize=5pt](2,2)(2.6,3)
	\psline[showpoints=true,dotsize=5pt](2.6,1)(3.2,0)
	\uput[r](2,2){$x_0$}
	\uput[r](2.6,1){$x_1$}
	\uput[l](1.4,1){$x_2$}
	\uput[u](1.4,3){$y_0$}
	\uput[u](2,3){$y_1$}
	\uput[u](2.6,3){$y_3$}
	\uput[dr](3.2,0){$y_2$}
	\uput[u](2,-1){$G_{310}$}
\end{pspicture}
}

\def\GTwoTwoZeroZero{
\begin{pspicture}(0,0)(3,3)
	\pspolygon[showpoints=true,dotsize=5pt](1,2)(2,2)(2,1)(1,1)
	\psline(2,2)(1,1)(1,2)(2,1)
	\psline[showpoints=true,dotsize=5pt](1,2)(1,3)
	\psline[showpoints=true,dotsize=5pt](1,2)(0.4,3)
	\psline[showpoints=true,dotsize=5pt](2,2)(2,3)
	\psline[showpoints=true,dotsize=5pt](2,2)(2.6,3)
	\uput[l](1,2){$x_0$}
	\uput[r](2,2){$x_1$}
	\uput[r](2,1){$x_2$}
	\uput[l](1,1){$x_3$}
	\uput[u](0.4,3){$y_0$}
	\uput[u](1,3){$y_2$}
	\uput[u](2,3){$y_1$}
	\uput[u](2.6,3){$y_3$}
	\uput[u](1.5,-1){$G_{2200}$}
\end{pspicture}
}

\def\GTwoTwoZero{
\begin{pspicture}(0,0)(4,3)
	\pspolygon[showpoints=true,dotsize=5pt](2,2)(1.4,1)(2.6,1)
	\psline[showpoints=true,dotsize=5pt](2,2)(1.7,3)
	\psline[showpoints=true,dotsize=5pt](2,2)(2.3,3)
	\psline[showpoints=true,dotsize=5pt](2.6,1)(3.0,-0.2)
	\psline[showpoints=true,dotsize=5pt](2.6,1)(3.4,0.2)
	\uput[r](2,2){$x_0$}
	\uput[r](2.6,1){$x_1$}
	\uput[l](1.4,1){$x_2$}
	\uput[u](1.7,3){$y_0$}
	\uput[u](2.3,3){$y_1$}
	\uput[r](3.4,+0.2){$y_2$}
	\uput[l](3.0,-0.2){$y_3$}
\end{pspicture}
}

\def\GTwoTwoZeroPrime{
\begin{pspicture}(0,0)(4,3)
	\pspolygon[showpoints=true,dotsize=5pt](2,2)(1.4,1)(2.6,1)
	\psline[showpoints=true,dotsize=5pt](2,2)(1.7,3)
	\psline[showpoints=true,dotsize=5pt](2,2)(2.3,3)
	\psline[showpoints=true,dotsize=5pt](2.6,1)(3.0,-0.2)
	\psline[showpoints=true,dotsize=5pt](2.6,1)(3.4,0.2)
	\psline[showpoints=true,dotsize=5pt](2,2)(3.2,2)(2.6,1)
	\uput[l](2,2){$x_0$}
	\uput[r](2.6,1){$x_1$}
	\uput[l](1.4,1){$x_2$}
	\uput[u](1.7,3){$y_0$}
	\uput[u](2.3,3){$y_1$}
	\uput[r](3.4,+0.2){$y_2$}
	\uput[l](3.0,-0.2){$y_3$}
	\uput[r](3.2,2){$z$}
\end{pspicture}
}

\newarray\ColoursK
\newarray\ColoursI
\newcommand\BaseGraph{
	\pspolygon[fillstyle=solid, fillcolor=black!20!white]
		(0,2)(-0.4,0)(0,-0.4)(2,0)(0.3,0.3)
 	\pscircle[fillstyle=solid, fillcolor=black!20!white](0,0){0.5}
 	\psline[showpoints=true](0,2)(2,0)
 	\uput[dl](0,2){$K_0$}
 	\uput[dr](2,0){$K_1$}
 	\uput[dl](-0.2,-0.2){$K_2$}
 	\psline[showpoints=true]
		(-0.3,0.3)(-2,2)(0,2)(2,2)(2,0)(2,-2)(0.3,-0.3)	
  	\psline[showpoints=true]
 		(-0.3,-0.3)(0,-2)(+0.3,-0.3)
 	\pscircle[linestyle=dashed]( 2,2){0.3} \rput{0}(2.4,2.4){$I_{0,1}$}
 	\pscircle[linestyle=dashed](2,-2){0.3} \rput{0}(2.6,-2){$I_{1,2}$}
 	\pscircle[linestyle=dashed](-2,2){0.3} \rput{0}(-2,2.5){$I_{2,0}$}
 	\pscircle[linestyle=dashed](0,-2){0.3} \rput{0}(-0.7,-2){$I_{2,2}$}
	\PrintColours{
		\rput{0}(0.85,0.85){$\ColoursK(1)$}
		\rput{0}(0.8,0){$\ColoursK(2)$}
		\rput{0}(0,0.8){$\ColoursK(3)$}
		\rput{0}(0,0){$\ColoursK(4)$}
		\rput{0}(1.3,1.8){$\ColoursI(1)$}
		\rput{0}(1.8,1.3){$\ColoursI(2)$}
		\rput{0}(1.8,-1){$\ColoursI(3)$}
		\rput{0}(1.3,-1){$\ColoursI(4)$}
		\rput{0}(-1,1.3){$\ColoursI(5)$}
 		\rput{0}(-1,1.8){$\ColoursI(6)$}
 		\rput{0}(-0.3,-1.4){$\ColoursI(7)$}
 		\rput{0}(+0.3,-1.4){$\ColoursI(8)$}
	}
}

\def\FigureColourGOneOneOne{
	\begin{pspicture}(-2.5,-2.5)(3.5,3.5)
		\psset{dotsize=3pt}
		\readarray{ColoursK}{2 & 0 & 1 & 1}
		\readarray{ColoursI}{0 & 1 & 1 & 2 & 2 & 0 & 0 & 2}
		\BaseGraph
	  	\psline[showpoints=true](0,2)(0,3) \uput[u](0,3){$y_0$}
	  	\psline[showpoints=true](2,0)(3,0) \uput[r](3,0){$y_1$}
	  	\psline[showpoints=true](-1.4,0)(-0.4,0) \uput[l](-1.4,0){$y_2$}
		\PrintColours{
			\rput{0}(0.2,2.5){$0$}
			\rput{0}(2.5,0.2){$1$}
			\rput{0}(-1,0.2){$2$}
		}
	\end{pspicture}
}

\def\FigureColourGFourZeroZero{
	\begin{pspicture}(-2.5,-2.5)(3.5,3.5)
	\psset{unit=1cm}
		\psset{dotsize=3pt}
		\readarray{ColoursK}{3 & 0 & 1 & 1}
		\readarray{ColoursI}{0 & 1 & 1 & 2 & 2 & 0 & 0 & 2}
		\BaseGraph
	  	\psline[showpoints=true](0,2)(-1.5,3) \uput[u](-1.5,3){$y_0$}
	  	\psline[showpoints=true](0,2)(-0.5,3) \uput[u](-0.5,3){$y_1$}
	  	\psline[showpoints=true](0,2)(+0.5,3) \uput[u](+0.5,3){$y_2$}
	  	\psline[showpoints=true](0,2)(+1.5,3) \uput[u](+1.5,3){$y_3$}
		\PrintColours{
			\rput{0}(-1.4,2.7){$0$}
			\rput{0}(-0.5,2.7){$1$}
			\rput{0}(+0.5,2.7){$2$}
			\rput{0}(+1.4,2.7){$3$}
		}
	\end{pspicture}
}

\def\FigureColourGThreeOneZero{
	\begin{pspicture}(-2.5,-2.5)(3.5,3.5)
	\psset{unit=1cm}
		\psset{dotsize=3pt}
		\readarray{ColoursK}{3 & 0 & 1 & 1}
		\readarray{ColoursI}{0 & 1 & 1 & 2 & 2 & 0 & 0 & 2}
		\BaseGraph
	  	\psline[showpoints=true](0,2)(-1.5,3) \uput[u](-1.5,3){$y_0$}
	  	\psline[showpoints=true](0,2)(-0.5,3) \uput[u](-0.5,3){$y_1$}
	  	\psline[showpoints=true](0,2)(+1.5,3) \uput[u](+1.5,3){$y_3$}
	  	\psline[showpoints=true](2,0)(3,0) \uput[r](3,0){$y_2$}
		\PrintColours{
			\rput{0}(-1.4,2.7){$0$}
			\rput{0}(-0.5,2.7){$1$}
			\rput{0}(+1.4,2.7){$3$}
			\rput{0}(2.5,0.2){$2$}
		}
	\end{pspicture}
}

\def\FigureColourGTwoTwoZeroZero{
	\begin{pspicture}(-2.5,-2.5)(3.5,3.5)
	\psset{unit=1cm}
		\psset{dotsize=3pt}
		\pspolygon[fillstyle=solid, fillcolor=black!20!white]
			(0,2)(-0.4,0)(0,-0.4)(2,0)(0.3,0.3)
		\pspolygon[fillstyle=solid, fillcolor=black!20!white]
	 		(+0.4,0)(2,2)(0,0.4)
	 	\pscircle[fillstyle=solid, fillcolor=black!20!white](0,0){0.5}
	 	\psline[showpoints=true](0,2)(2,2)(2,0)
		\psarc(0,0){2}{0}{90}
	 	\uput[ul](0,2){$K_0$}
	 	\uput[ur](2,2){$K_1$}
	 	\uput[dr](2,0){$K_2$}
	 	\uput[d](0,-0.4){$K_3$}
	 	\psline[showpoints=true]
			(-0.35,0.2)(-2,1)(0,2)(1,4)(2,2)(4,1)(2,0)(1,-2)(0.2,-0.35)	
	  	\psline[showpoints=true]
	 		(-0.4,0)(-1.4,-1.4)(0,-0.4)
 	  	\psline[showpoints=true]{<->}%
			(3,3.5)(4,4)(3.5,3)
 	  	\psline[showpoints=true]{<->}%
			(3.5,-1)(4,-2)(3,-1.5)
	 	\pscircle[linestyle=dashed]( 1,4){0.3} \rput{0}(1.6,4){$I_{0,1}$}
	 	\pscircle[linestyle=dashed]( 4,1){0.3} \rput{0}(4,0.4){$I_{1,2}$}
	 	\pscircle[linestyle=dashed](1,-2){0.3} \rput{0}(1.7,-2){$I_{2,3}$}
	 	\pscircle[linestyle=dashed](-2,1){0.3} \rput{0}(-2,0.4){$I_{3,0}$}
	 	\pscircle[linestyle=dashed](-1.4,-1.4){0.3} \rput{0}(-1.4,-2){$I_{3,3}$}
	 	\pscircle[linestyle=dashed](4,4){0.3} \rput{0}(4.6,4){$I_{0,2}$}
	 	\pscircle[linestyle=dashed](4,-2){0.3} \rput{0}(4.6,-2){$I_{1,3}$}
	  	\psline[showpoints=true](0,2)(-1.5,2) \uput[l](-1.5,2){$y_0$}
	  	\psline[showpoints=true](0,2)(0,3.5) \uput[l](0,3.5){$y_2$}
	  	\psline[showpoints=true](2,2)(2,3.5) \uput[r](2,3.5){$y_1$}
	  	\psline[showpoints=true](2,2)(3.5,2) \uput[r](3.5,2){$y_3$}
		\readarray{ColoursK}{2 & 0 & 0 & 1 & 0 & 3 & 1}
		\readarray{ColoursI}{0 & 1 & 1 & 2 & 2 & 3 & 3 & 0 & 0 & 2 & 1 & 3 & 0 & 3}
		\PrintColours{
			\rput{0}(1.0,2.2){$\ColoursK(1)$}
			\rput{0}(2.2,1.0){$\ColoursK(2)$}
			\rput{0}(1,0){$\ColoursK(3)$}
			\rput{0}(0,1.0){$\ColoursK(4)$}
			\rput{0}(1.0,1.0){$\ColoursK(5)$}
			\rput{0}(1.7,0.7){$\ColoursK(6)$}
			\rput{0}(0,0){$\ColoursK(7)$}
			\uput[l](1,3){$\ColoursI(1)$}
			\uput[r](1,3){$\ColoursI(2)$}
			\uput[u](3,1){$\ColoursI(3)$}
			\uput[d](3,1){$\ColoursI(4)$}
			\uput[r](1,-1){$\ColoursI(5)$}
			\uput[l](1,-1){$\ColoursI(6)$}
			\uput[d](-1,1){$\ColoursI(7)$}
			\uput[u](-1,1){$\ColoursI(8)$}
			\uput[ul](3.6,3.6){$\ColoursI(9)$}
			\uput[dr](3.6,3.6){$\ColoursI(10)$}
			\uput[ur](3.6,-1.6){$\ColoursI(11)$}
			\uput[dl](3.6,-1.6){$\ColoursI(12)$}
			\uput[dr](-1,-1){$\ColoursI(13)$}
			\uput[ul](-1,-1){$\ColoursI(14)$}
			\rput{0}(-1.0,2.2){$0$}
			\rput{0}(-0.2,3.0){$2$}
			\rput{0}(+2.2,3.0){$1$}
			\rput{0}(+3.0,2.2){$3$}
		}
	\end{pspicture}
}

\def\FigureColourGTwoTwoZeroPrime{
	\begin{pspicture}(-2.5,-2.5)(3.5,3.5)
		\psset{dotsize=3pt}
		\readarray{ColoursK}{0 & 1 & 3}
		\readarray{ColoursI}{1 & 2 & 3 & 2 & 1 & 0}
	 	\pspolygon[showpoints=true](0,2)(2,0)(0,0)
	 	\uput[dl](0,2){$x_0$}
	 	\uput[dl](2,0){$x_1$}
	 	\uput[dl](0,0){$x_2$}
	 	\psline[showpoints=true]
			(0,0)(-2,2)(0,2)(2,2)(2,0)(2,-2)(0,0)	
	 	\pscircle[linestyle=dashed]( 2,2){0.3} \rput{0}(2.2,2.4){$\{z\} \cup I_{0,1}$}
	 	\pscircle[linestyle=dashed](2,-2){0.3} \rput{0}(2.6,-2){$I_{1,2}$}
	 	\pscircle[linestyle=dashed](-2,2){0.3} \rput{0}(-2,2.5){$I_{2,0}$}
		\PrintColours{
			\rput{0}(0.85,0.85){$\ColoursK(1)$}
			\rput{0}(0.85,0.2){$\ColoursK(2)$}
			\rput{0}(0.2,0.85){$\ColoursK(3)$}
			\rput{0}(1.3,1.8){$\ColoursI(1)$}
			\rput{0}(1.8,1.3){$\ColoursI(2)$}
			\rput{0}(1.8,-1){$\ColoursI(3)$}
			\rput{0}(1.3,-1){$\ColoursI(4)$}
			\rput{0}(-1,1.3){$\ColoursI(5)$}
	 		\rput{0}(-1,1.8){$\ColoursI(6)$}
		}
	  	\psline[showpoints=true](0,2)(-0.5,3) \uput[u](-0.5,3){$y_0$}
	  	\psline[showpoints=true](0,2)(+0.5,3) \uput[u](+0.5,3){$y_1$}
	  	\psline[showpoints=true](2,0)(3,+0.5) \uput[r](3,+0.5){$y_2$}
	  	\psline[showpoints=true](2,0)(3,-0.5) \uput[r](3,-0.5){$y_3$}
		\PrintColours{
			\rput{0}(-0.2,2.7){$0$}
			\rput{0}(+0.2,2.7){$1$}
			\rput{0}(2.7,+0.2){$2$}
			\rput{0}(2.7,-0.2){$3$}
		}
	\end{pspicture}
}

\begin{document}

\maketitle

\begin{abstract}

A {\em rainbow path} in an edge coloured graph is a path in which no two edges are coloured the same. A {\em rainbow colouring} of a connected graph $G$ is a colouring of the edges of $G$ such that every pair of vertices in $G$ is connected by at least one rainbow path. The minimum number of colours required to rainbow colour $G$ is called its {\em rainbow connection number}. Between them, Chakraborty et al. [J. Comb. Optim., 2011] and Ananth et al. [FSTTCS, 2012] have shown that for every integer $k$, $k \geq 2$, it is \NP-complete to decide whether a given graph can be rainbow coloured using $k$ colours.

A {\em split graph} is a graph whose vertex set can be partitioned into a clique and an independent set. Chandran and Rajendraprasad have shown that the problem of deciding whether a given split graph $G$ can be rainbow coloured using $3$ colours is \NP-complete and further have described a linear time algorithm to rainbow colour any split graph using at most one colour more than the optimum [COCOON, 2012]. In this article, we settle the computational complexity of the problem on split graphs and thereby discover an interesting dichotomy. Specifically, we show that the problem of deciding whether a given split graph can be rainbow coloured using $k$ colours is \NP-complete for $k \in \{2,3\}$, but can be solved in polynomial time for all other values of $k$.
\end{abstract}


%% Introduction

\section{Introduction}

An {\em edge colouring} of a graph is a function from its edge set to the set of natural numbers. A path in an edge coloured graph with no two edges sharing the same colour is called a {\em rainbow path}. An edge coloured graph is said to be {\em rainbow connected} if every pair of vertices is connected by at least one rainbow path. Such a colouring is called a {\em rainbow colouring} of the graph. A rainbow colouring using minimum possible number of colours is called {\em optimal}. The minimum number of colours required to rainbow colour a connected graph $G$ is called its {\em rainbow connection number}, denoted by $rc(G)$. For example, the rainbow connection number of a complete graph is $1$, that of a path is its length, that of an even cycle is half its length, and that of a tree is its number of edges. Note that disconnected graphs cannot be rainbow coloured and hence their rainbow connection number is left undefined. Any connected graph can be rainbow coloured by giving distinct colours to the edges of a spanning tree of the graph. Hence the rainbow connection number of any connected graph is less than its number of vertices. It is trivial to see that that $rc(G)$ is at least the diameter of $G$. It is easy to see that no two bridges in a graph can get the same colour under a rainbow colouring and hence $rc(G)$ is lower bounded by the number of bridges in the $G$.

The concept of rainbow colouring was introduced by Chartrand, Johns, McKeon, and Zhang  in \cite{chartrand2008rainbow} where they also determined the precise values of rainbow connection number for some special graphs. Subsequently, there have been various investigations towards finding good upper bounds for rainbow connection number in terms of other graph parameters \cite{caro2008rainbow,schiermeyer2009rainbow,chandran2011raindom,basavaraju2012radius} and for many special graph classes \cite{li2011linegraphs,chandran2011raindom,basavaraju2011products}. Behaviour of rainbow connection number in random graphs is also well studied \cite{caro2008rainbow,he2010rainthreshold,shang2011randombipartite,frieze2012rainbow}. A basic introduction to the topic can be found in Chapter $11$ of the book {\em Chromatic Graph Theory} by Chartrand and Zhang \cite{chartrand2008chromatic} and a survey of most of the recent results in the area can be found in the article by Li and Sun \cite{li2012rainsurvey} and also in their monograph {\em Rainbow Connection of Graphs} \cite{li2012rainbowbook}. 

\section{Our contribution}

In this article we focus on the computational complexity of the following decision problem on split graphs (Definition \ref{defClasses}).

\begin{problem}[$\rc(G,k)$]
\label{problemRainbowColour}
Given a connected graph $G$ and a positive integer $k$, decide whether $G$ can be rainbow coloured using $k$ colours.
\end{problem}
  
The first result showing the computational complexity of the above problem was due to Chakraborty, Fischer, Matsliah, and Yuster \cite{chakraborty2011hardness}. They showed that  it is \NP-hard to compute the rainbow connection number of an arbitrary graph. In particular, it was shown that the problem $\rc(G,2)$ is \NP-complete. Later, Ananth, Nasre, and Sarpatwar \cite{ananth2011fstrcs} complemented the above result and now we know that for every integer $k$, $k \geq 2$, the problem $\rc(G, k)$ is \NP-complete. This prompts one to look at the computational complexity of the problem on special graph classes. Chandran and Rajendraprasad have shown that $\rc(G,k)$ is solvable in linear time for threshold graphs, \NP-complete on split graphs for $k=3$ and \NP-complete on chordal graphs for all $k \geq 3$ \cite{chandran2012split}. It is easy to see that complete graphs alone can be rainbow coloured using $1$ colour. The complexity of the problem $\rc(G,k)$ on chordal graphs for $k = 2$ and split graphs for all positive integers $k$ except $1$ and $3$ was left open. In this article, we solve the same and thereby discover the following interesting dichotomy. 

\begin{theorem}
\label{theoremSplitDichotomy}
The problem $\rc(G,k)$ on split graphs is \NP-complete for $k \in \{2,3\}$ and polynomial-time solvable for all other values of $k$. 
\end{theorem}

% On the positive side, Basavaraju et al. have demonstrated an $O(nm)$-time $(r+3)$-factor approximation algorithm for rainbow colouring any graph with radius $r$ \cite{basavaraju2012radius}. Constant factor approximation algorithms for rainbow colouring Cartesian, strong and lexicographic products of non-trivial graphs are reported in \cite{basavaraju2011products}. Constant factor approximation algorithms for bridgeless chordal graphs, and additive approximation algorithms for interval, AT-free, threshold and circular arc graphs without pendant vertices will follow from the proofs of their upper bounds \cite{chandran2011raindom}.  To the best of our knowledge, no efficient optimal rainbow colouring algorithm has been reported for any non-trivial subclass of graphs.

\section{On the proofs}

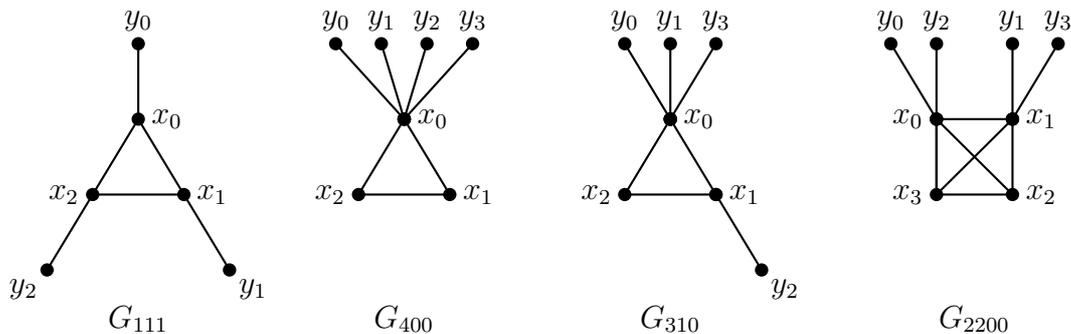
\begin{figure}
\begin{center}
\begin{pspicture}(0,0)(15,4)
\psset{unit=1cm}
%   	\psgrid
	\rput{0}( 2.0,2){\GOneOneOne}
	\rput{0}( 5.5,2){\GFourZeroZero}
	\rput{0}( 9.0,2){\GThreeOneZero}
	\rput{0}(13.0,2){\GTwoTwoZeroZero}
\end{pspicture}
\end{center}
\caption{Four special split graphs which constitute the set $\G$}
\label{figSpecialGraphs}
\end{figure}

First we show that the problem $\rc(G,k)$ is polynomial time solvable for $k \geq 4$ by demonstrating the following structural result whose proof is given in Appendix \ref{secFourColours}. Let $\pen(G)$ denotes the set of pendant vertices (vertices with exactly one neighbour) in a graph $G$.

\begin{lemma}
\label{lemmaSplitPendant}
If a split graph $G$, under some isomorphism, contains any of the graphs $H \in \G$ in Figure \ref{figSpecialGraphs} as a subgraph with $\pen(H) \subseteq \pen(G)$, then $rc(G) = |\pen(G)|$.
\end{lemma}

From the above lemma and the easy observation that $rc(G) \geq |\pen(G)|$ for any graph, it follows that for each integer $k \geq 4$ there exists a polynomial time algorithm to check if the rainbow connection number of a split graph is at most $k$. The proof gives an explicit rainbow colouring of $G$ using $|\pen(G)|$ colours if it contains any of the graphs $H \in \G$ as a subgraph with $\pen(H) \subset \pen(G)$, and thus we show that any split graph with rainbow connection number at least $4$ can be optimally rainbow coloured in polynomial time (Corollary \ref{corSplitPoly} in Appendix \ref{secFourColours}).

Next we show that the problem $\rc(G,2)$ remains \NP-complete for split graphs. This is established by showing a two-step reduction. Given a graph $G = (V,E)$, and a collection of subsets $\S$ of $V$, the problem $\bcc(G, \S)$ is to decide whether there exists a {\em bipartitioning function} $X : \S \into 2^V$ such that $X(T) \subset T, \forall T \in \S$ and $G$ is covered by the family of bicliques $\{ (X(T), T \setminus X(T)) : T \in \S\}$. We show that $\sat$ is reducible to $\bcc$ which in turn is reducible to $\rc(G,2)$ with $G$ being a split graph (Lemmata \ref{lemmaBCCtoRC} and \ref{lemma3SATtoBCC} in Appendix \ref{secTwoColours}).

\section{Consequences}

The problems below are only superficially different from the $\rc(G,2)$ problem on split graphs (see the discussion after Problem \ref{problemBipartiteRainbow} in Appendix \ref{secTwoColours}) and hence we deduce that they are also \NP-complete (the problem size being $O(mn)$ in each case).

\begin{problem}[$\textsc{EnsureDistinctRows}(C)$]
\label{problemMatrix}
Given a subset $C \subset [m] \times [n]$ of locations, decide whether there exists an $m \times n$ matrix $M$ with entries from $\{0,1\}$ such that any two rows of $M$ will remain distinct, no matter what changes are made to the entries of $M$ at locations in $C$.
\end{problem}

\begin{problem}[$\textsc{OrthogonalPacking}(B)$]
\label{problemPacking}
Given a set $B$ of $m$ $n$-dimensional boxes whose sides are either $1$ or $1/2$ in each dimension, decide whether they can be packed without rotation into an $n$-dimensional unit cube.
\end{problem}

We would also like to emphasise that the problem $\rc(G,2)$ is known to be linear time solvable for threshold graphs, which are split graphs in which the neighbourhoods of the independent set vertices form a total order under inclusion. In particular a threshold graph $G$ can be rainbow coloured using $2$ colours if and only if the degrees of the vertices in a maximum independent set $I$ of $G$ satisfy the Kraft's inequality, viz. $\sum_{v \in I} 2^{-d(v)} \leq 1$ where $d(v)$ denotes the degree of a vertex $v$ \cite{chandran2012split}. The problem $\textsc{EnsureDistinctRows(C)}$ can be viewed as a combinatorial generalisation of the problem of constructing a prefix-free code given a set of desired lengths. The latter is poly-time solvable while the above generalisation is shown here to be \NP-complete.

\bibliographystyle{plain}
\bibliography{deepak}

\begin{thebibliography}{10}

\bibitem{ananth2011fstrcs}
Prabhanjan Ananth, Meghana Nasre, and Kanthi~K. Sarpatwar.
\newblock Rainbow connectivity: {H}ardness and tractability.
\newblock In {\em FSTTCS 2012}, volume~13, pages 241--251, 2011.

\bibitem{basavaraju2011products}
Manu Basavaraju, L.~Sunil Chandran, Deepak Rajendraprasad, and Arunselvan
  Ramaswamy.
\newblock Rainbow connection number of graph power and graph products.
\newblock {\em Accepted for publication in Graphs and Combinatorics. Preprint:
  arXiv:1104.4190v2 [math.CO]}, 2011.

\bibitem{basavaraju2012radius}
Manu Basavaraju, L.~Sunil Chandran, Deepak Rajendraprasad, and Arunselvan
  Ramaswamy.
\newblock Rainbow connection number and radius.
\newblock {\em Graphs and Combinatorics}, pages 1 -- 11, 2012.

\bibitem{caro2008rainbow}
Yair Caro, Arie Lev, Yehuda Roditty, Zsolt Tuza, and Raphael Yuster.
\newblock On rainbow connection.
\newblock {\em Electron. J. Combin.}, 15(1):Research paper 57, 13, 2008.

\bibitem{chakraborty2011hardness}
Sourav Chakraborty, Eldar Fischer, Arie Matsliah, and Raphael Yuster.
\newblock Hardness and algorithms for rainbow connection.
\newblock {\em Journal of Combinatorial Optimization}, 21(3):330--347, 2011.

\bibitem{chandran2012split}
L.~Sunil Chandran and Deepak Rajendraprasad.
\newblock Rainbow colouring of split and threshold graphs.
\newblock In Joachim Gudmundsson, Julián Mestre, and Taso Viglas, editors,
  {\em Computing and Combinatorics}, volume 7434 of {\em Lecture Notes in
  Computer Science}, pages 181--192. Springer Berlin / Heidelberg, 2012.
\newblock 10.1007/978-3-642-32241-9\_16.

\bibitem{chartrand2008rainbow}
Gary Chartrand, Garry~L. Johns, Kathleen~A. McKeon, and Ping Zhang.
\newblock Rainbow connection in graphs.
\newblock {\em Math. Bohem.}, 133(1):85--98, 2008.

\bibitem{chartrand2008chromatic}
Gary Chartrand and Ping Zhang.
\newblock {\em {Chromatic Graph Theory}}.
\newblock Chapman \& Hall, 2008.

\bibitem{frieze2012rainbow}
A.~Frieze and C.E. Tsourakakis.
\newblock Rainbow connectivity of {$G(n,p)$} at the connectivity threshold.
\newblock {\em Preprint: arXiv:1201.4603}, 2012.

\bibitem{he2010rainthreshold}
Jing He and Hongyu Liang.
\newblock On rainbow-$k$-connectivity of random graphs.
\newblock {\em Information Processing Letters}, 112(10):406--410, 2012.

\bibitem{li2012rainsurvey}
Xueliang Li, Yongtang Shi, and Yuefang Sun.
\newblock Rainbow connections of graphs: A survey.
\newblock {\em Graphs and Combinatorics}, pages 1--38, 2012.
\newblock 10.1007/s00373-012-1243-2.

\bibitem{li2011linegraphs}
Xueliang Li and Yuefang Sun.
\newblock Upper bounds for the rainbow connection numbers of line graphs.
\newblock {\em Graphs and Combinatorics}, pages 1--13, 2011.
\newblock 10.1007/s00373-011-1034-1.

\bibitem{li2012rainbowbook}
Xueliang Li and Yuefang Sun.
\newblock {\em {Rainbow Connections of Graphs}}.
\newblock Springerbriefs in Mathematics. Springer, 2012.

\bibitem{schiermeyer2009rainbow}
Ingo Schiermeyer.
\newblock Rainbow connection in graphs with minimum degree three.
\newblock In {\em Combinatorial Algorithms}, volume 5874 of {\em Lecture Notes
  in Comput. Sci.}, pages 432--437. Springer, Berlin, 2009.

\bibitem{shang2011randombipartite}
Yilun Shang.
\newblock A sharp threshold for rainbow connection of random bipartite graphs.
\newblock {\em Int. J. Appl. Math.}, 24(1):149--153, 2011.

\bibitem{chandran2011raindom}
L.~Sunil~Chandran, Anita Das, Deepak Rajendraprasad, and Nithin~M. Varma.
\newblock Rainbow connection number and connected dominating sets.
\newblock {\em Journal of Graph Theory}, 71(2):206 -- 218, 2012.

\end{thebibliography}

\clearpage
\appendix

\section{Appendix}

\subsection{Notation and definitions}

All graphs considered in this article are finite, simple and undirected. For a graph $G$, we use $V(G)$ and $E(G)$ to denote its vertex set and edge set respectively. An edge $\{u,v\}$ in a graph may be denoted by $uv$ to reduce clutter. Unless mentioned otherwise, $n$ and $m$ will respectively denote the number of vertices and edges of the graph in consideration. The subgraph of $G$ induced on a vertex set $S \subset V(G)$ is denoted by $G[S]$.

The shorthand $[n]$ denotes the set $\{1, \ldots, n\}$. The cardinality of a set $S$ is denoted by $|S|$ and the family of all subsets of $S$ is denoted by $2^S$. The union of two disjoint sets $A$ and $B$ is denoted by $A \dot \cup B$.

\begin{definition}
Let $G$ be a connected graph. The {\em length} of a path is its number of edges. The {\em distance} between two vertices $u$ and $v$ in $G$, denoted by $d(u,v)$ is the length of a shortest path between them in $G$. The {\em diameter} of $G$ is $diam(G) := \max_{u,v \in V(G)}{d(u,v)}$.
\end{definition}

\begin{definition} \label{defDegree}
The {\em neighbourhood} $N(v)$ of a vertex $v$ is the set of vertices adjacent to $v$ but not including $v$.  A vertex is called {\em pendant} if its degree is $1$. An edge incident on a pendant vertex is called a {\em pendant edge} and the set of pendant vertices of a graph $G$ is denoted by $\pen(G)$.
\end{definition}

\begin{definition} \label{defClasses}
A graph $G$ is called {\em chordal}, if there is no induced cycle of length greater than $3$. A graph $G$ is a {\em split graph}, if $V(G)$ can be partitioned into a clique and an independent set. A graph $G$ is a {\em threshold graph}, if there exists a weight function $w:V(G) \into \R$ and a real constant $t$ such that two vertices $u, v \in V(G)$ are adjacent if and only if $w(u) + w(v) \geq t$.
\end{definition}

% Before getting into the main results, we note two elementary and well known observations on rainbow colouring whose proofs we omit. 
% 
% \begin{observation}
% \label{obsDiameter}
% For every connected graph $G$, we have $rc(G) \geq diam(G)$.
% \end{observation}
% \begin{proof}
% If $diam(G) = d$, then there exists two vertices $u, v \in V(G)$ which are a distance $d$ apart i.e, every path between $u$ and $v$ has length at least $d$. Since we need to have a rainbow path between $u$ and $v$ we require at least $d$ colours in any rainbow colouring of $G$.
% \end{proof}

% \begin{observation}
% \label{obsPendant}
% If $u$ and $v$ are two pendant vertices in a connected graph $G$, then their incident edges get different colours in any rainbow colouring of $G$. In particular, if $G$ has $p$ pendant vertices, then $rc(G) \geq p$. 
% \end{observation}
% \begin{proof}
% If the edges incident on $u$ and $v$ share the same colour, then there is no rainbow path between $u$ and $v$. In any edge colouring of $G$ that uses fewer than $p$ colours, there exists two pendant vertices $u$ and $v$ such that their incident edges the same colour. Hence the colouring cannot be a rainbow colouring of $G$.
% \end{proof}

\subsection{More than three colours: Polynomial time solution}
\label{secFourColours}

\subsubsection*{Proof of Lemma \ref{lemmaSplitPendant}}

\begin{figure}
\begin{center}
\begin{pspicture}(0,0)(15,16)
%  \psgrid
	\rput{0}( 3,13){\FigureColourGOneOneOne} 	
	\rput{0}( 3,9.5){Case \ref{caseGOneOneOne} (Equation \ref{eqnColourGOneOneOne})}
	\rput{0}(11,13){\FigureColourGFourZeroZero} 
	\rput{0}(11,9.5){Case \ref{caseGFourZeroZero} (Equation \ref{eqnColourGFourZeroZero})}
	\rput{0}( 3,4){\FigureColourGThreeOneZero} 	
	\rput{0}( 3,0.5){Case \ref{caseGThreeOneZero} (Equation \ref{eqnColourGThreeOneZero})}
	\rput{0}(11,4){\FigureColourGTwoTwoZeroZero} 
	\rput{0}(11,0.5){Case \ref{caseGTwoTwoZeroZero} (Equation \ref{eqnColourGTwoTwoZeroZero})}
\end{pspicture}
\end{center}
\caption{
A partial illustration of the four colourings defined in the proof of Lemma \ref{lemmaSplitPendant}. Only one representative vertex from the each part of the independent set is illustrated. The edge-colours are indicated in red to distinguish them from other labels.
} 
\label{figColourSpecialGraphs}
\end{figure}

\begin{proof}[Statement]
If a split graph $G$, under some isomorphism, contains any of the graphs $H \in \G$ in Figure \ref{figSpecialGraphs} as a subgraph with $\pen(H) \subseteq \pen(G)$, then $rc(G) = |\pen(G)|$.

Let us relabel the vertices of $G$ so that $H$ is contained as a (labelled) subgraph of $G$ with $\pen(H) \subset \pen(G)$. First we note that it suffices to prove the statement when $\pen(G) = \pen(H)$. Suppose $P' = \pen(G) \setminus \pen(H)$ is non-empty. Then consider the induced subgraph $G'$ of $G$ obtained by removing all the vertices in $P'$. Note that $G'$ also has $H$ as a subgraph with $\pen(H) \subset \pen(G')$. If $G'$ can be rainbow coloured with $|\pen(G')|$ colours, we can easily extend this to a rainbow colouring of $G$ with $p$ colours by giving a new colour to each edge of $G$ incident to a vertex in $P'$. Henceforth in this proof we assume $\pen(G) = \pen(H)$.

The proof is divided into four cases based on $H \in \G$. In each case, we describe an edge-colouring $c_G$ of $G$ using $|\pen(H)|$ colours and then show that $c_G$ makes $G$ rainbow connected. A partial illustration of the colourings is given in Figure \ref{figColourSpecialGraphs}. In each case, we set $K$ to be a maximal clique in $G$, $I = V(G) \setminus K$, $P = \pen(G)$ and $I' = I \setminus P$. For each $v \in I'$, we can assume that $v$ has exactly $2$ neighbours in $K$. Remaining edges from $I'$ to $K$ are not used in our colouring and hence may be assumed absent. In the first three cases below, that is when $H \in \{G_{111}, G_{400}, G_{310}\}$, we partition $K$ and $I'$ as follows. Vertices in $K$ are grouped into $3$ parts $K_0 = \{x_0\}$, $K_1 = \{x_1\}$ and $K_2 = K \setminus \{x_0, x_1\}$ while the vertices in $I'$ are grouped into $4$ parts $I_{0,1}$, $I_{1,2}$, $I_{2,0}$ and $I_{2,2}$, where $I_{i,j}$, $i \neq j$, consists of those vertices in $I'$ with one neighbour in $K_i$ and one neighbour in $K_j$ and $I_{2,2}$ consists of those vertices in $I'$ with both neighbours in $K_2$. In the fourth case, $K$ is partitioned into $4$ parts $K_i = \{x_i\}, i \in \{0,1,2\}$, and $K_3 = K \setminus \{x_0, x_1, x_2\}$ while $I'$ is partitioned into $7$ parts $I_{i,j}, \{i, j\} \subset \Z_4, i \neq j$, and $I_{3,3}$ as before. While defining a colouring $c_G$ of $E(G)$, we will use the shorthand $c_G(A,B) = i$ to indicate that $c_G(\{a,b\}) = i$, for all $\{a,b\} \in E(G)$ such that $a \in A$ and $b \in B$.

\begin{case}[$H = G_{111}$]
\label{caseGOneOneOne}
In this case,  $\pen(G) = \{y_0, y_1, y_2\}$ and thus $p=3$. We define the $3$-colouring $c_G : E(G) \into \Z_3$ (See Figure \ref{figColourSpecialGraphs}). 

\begin{equation}
\begin{array}{rcll}
c_G(K_i, K_j) 					&=& k,	& \textnormal{where } \{i,j,k\} = \Z_3, \\
c_G(K_2, K_2) 					&=& 1, 	& \\
c_G(K_i, I \setminus I_{2,2}) 	&=& i, 	& \forall i \in \Z_3, \textnormal{ and} \\
c_G(\{v, u_l\}) 				&=& l, 	& \forall v \in I_{2,2}, l \in \{0,2\} \textnormal{ and } N(v) = \{u_0, u_2\}. 
\end{array}
\label{eqnColourGOneOneOne}
\end{equation}

Now we show that $c_G$ is a rainbow colouring of $G$ by listing down a rainbow path between every pair of vertices which are at a distance of at least $2$ from each other. Let $I_i$, $i \in \{0,1,2\}$, denote the set of vertices in $I$ with at least one neighbour in $K_i$. Note that a vertex in $I_{i,j}$ is part of $I_i$ and $I_j$ and hence two distinct vertices in $I_{i,j}$ are connected by a rainbow path of the second type in the list below.

\begin{align*}
\Path{u \in I_i}{v \in K_j, \; i \neq j}
	&: u \Edge{i} K_i \Edge{k} K_j & 
	(\textnormal{where } \{i,j,k\} = \Z_3) \\
\Path{u \in I_i}{v \in I_j, \; i \neq j}
	&: u \Edge{i} K_i \Edge{k} K_j \Edge{j} v 
	& (\textnormal{where } \{i,j,k\} = \Z_3) \\
\Path{u \in I_{2,2} \cup \{y_2\}}{v \in K_2, \; v \notin N(u)}
	&: u \Edge{2} K_2 \Edge{1} v & \\
\Path{u \in I_{2,2} \cup \{y_2\}}{v \in I_{2,2}, \; v \neq u}
	&: u \Edge{2} K_2 \Edge{1} K_2 \Edge{0} v &
\end{align*}
\end{case}

\begin{case}[$H = G_{400}$]
\label{caseGFourZeroZero}
In this and next two cases, $\pen(G) = \{y_0, \ldots, y_3\}$ and thus $p = 4$. We define a $4$-colouring $c_G : E(G) \into \Z_4$. 

\begin{equation}
\begin{array}{rcll}
c_G(K_0, K_1) 					&=& 3, & \\
c_G(K_1, K_2) 					&=& 0, & \\
c_G(K_2, K_0) 					&=& 1, & \\
c_G(K_2, K_2) 					&=& 1, & \\
c_G(\{y_i, x_0\}) 				&=& i, & \forall i \in \Z_4, \\
c_G(K_i, I' \setminus I_{2,2}) 	&=& i, & \forall i \in \{0,1,2\},  \textnormal{ and} \\
c_G(\{v, u_l\}) 				&=& l, & \forall v \in I_{2,2}, l \in \{0,2\} 
										 \textnormal{ and } N(v) = \{u_0, u_2\}.
\end{array}
\label{eqnColourGFourZeroZero}
\end{equation}

Notice that the colouring defined by Equation \ref{eqnColourGFourZeroZero} is similar to that defined by Equation \ref{eqnColourGOneOneOne} except for the pendant edges and the clique edge $\{x_0, x_1\}$. Now we show that $c_G$ is a rainbow colouring of $G$ by listing down a rainbow path between every pair of vertices which are at a distance of at least $2$ from each other. This time, let $I_i$, $i \in \{0,1,2\}$, denote the set of vertices in $I \setminus \{y_1,y_2,y_3\}$ with at least one neighbour in $K_i$.

\begin{align*}
\Path{I_i}{K_j \then I_j, \; i \neq j}
	&: I_i \Edge{i} K_i \Edge{k} K_j \Edge{j} I_j \; 
	(\textnormal{where }k \in \Z_4 \setminus \{i,j\})\\
\Path{I_{2,2}}{K_2 \then I_{2,2}}
	&: I_{2,2} \Edge{2} K_2 \Edge{1} K_2 \Edge{0} I_{2,2} \\
\Path{\{y_1,y_2,y_3\}}{v \in I_{2,0} \cup I_{0,1}}
	&: y_i \Edge{i} K_0 \Edge{0} v \\
\Path{y_1}{K_1 \then K_2 \then I_{2,2} \cup I_{1,2}}
	&: y_1 \Edge{1} K_0 \Edge{3} K_1 \Edge{0} K_2 \Edge{2} I_{2,2} \cup I_{1,2} \\
\Path{y_2}{K_2 \then I_{2,2}}
	&: y_2 \Edge{2} K_0 \Edge{1} K_2 \Edge{0} I_{2,2} \\
\Path{y_2}{K_1 \then I_{1,2}}
	&: y_2 \Edge{2} K_0 \Edge{3} K_1 \Edge{1} I_{1,2} \\
\Path{y_3}{K_2 \then I_{2,2} \cup I_{1,2}}
	&: y_3 \Edge{3} K_0 \Edge{1} K_2 \Edge{2} I_{2,2} \cup I_{1,2} \\
\Path{y_3}{K_1}
	&: y_3 \Edge{3} K_0 \Edge{1} K_2 \Edge{0} K_1
\end{align*}

\end{case}

\begin{case}[$H = G_{310}$]
\label{caseGThreeOneZero}
The colouring $c_G: E(G) \into \Z_4$ that we define in this case is similar to Case \ref{caseGFourZeroZero}. The only difference is that the pendant vertex $y_2$ is now adjacent to $x_1$ instead of $x_0$.

\begin{equation}
\begin{array}{rcll}
c_G(K_0, K_1) 					&=& 3, & \\
c_G(K_1, K_2) 					&=& 0, & \\
c_G(K_2, K_0) 					&=& 1, & \\
c_G(K_2, K_2) 					&=& 1, & \\
c_G(\{y_i, x_0\}) 				&=& i, & \forall i \in \{0,1,3\}, \\
c_G(\{y_2, x_1\}) 				&=& 2, & \\
c_G(K_i, I' \setminus I_{2,2}) 	&=& i, & \forall i \in \{0,1,2\}, \textnormal{ and} \\
c_G(\{v, u_l\}) 				&=& l, & \forall v \in I_{2,2}, l \in \{0,2\} 
										 \textnormal{ and } N(v) = \{u_0, u_2\}.
\end{array}
\label{eqnColourGThreeOneZero}
\end{equation}

Since all pairs of vertices not involving $y_2$ are connected by rainbow paths as described in Case \ref{caseGFourZeroZero}, we only indicate below rainbow paths from $y_2$ to every other vertex in $G$ to claim that $c_G$ rainbow connects $G$. 

\begin{align*}
\Path{y_2}{v \in I_{0,1} \cup I_{1,2}}
	&: y_2 \Edge{2} K_1 \Edge{1} v \\
\Path{y_2}{K_0 \then K_2 \then I_{2,2}}
	&: y_2 \Edge{2} K_1 \Edge{3} K_0 \Edge{1} K_2 \Edge{0} I_{2,2} \\
\Path{y_2}{v \in \{y_0\} \cup I_{2,0}}
	&: y_2 \Edge{2} K_1 \Edge{3} K_0 \Edge{0} v \\
\Path{y_2}{y_1}
	&: y_2 \Edge{2} K_1 \Edge{3} K_0 \Edge{1} y_1 \\
\Path{y_2}{y_3}
	&: y_2 \Edge{2} K_1 \Edge{0} K_2 \Edge{1} K_0 \Edge{3} y_3 
\end{align*}
\end{case}

\begin{case}[$H = G_{2200}$]
\label{caseGTwoTwoZeroZero}
Recall that in this case, unlike the previous three cases, we have partitioned $K$ into $4$ parts and $I'$ into $7$ parts. The colouring $c_G : E(G) \into \Z_4$ is defined as follows (See Figure \ref{figColourSpecialGraphs}).

\begin{equation}
\begin{array}{rcll}
c_G(K_i, K_{i+1}) 				&=& i+2, 	& i \in \{0,2,3\}, \\
c_G(K_1, K_{2}) 				&=& 0, 		& \\
c_G(K_i, K_{i+2}) 				&=& i+3, 	& i \in \{0,1\}, \\
c_G(K_3, K_3) 					&=& 1, 		& \\
c_G(K_i, I' \setminus I_{3,3}) 	&=& i, 		& \forall i \in \Z_4, \\
c_G(\{v, u_l\}) 				&=& l, 		& \forall v \in I_{3,3}, l \in \{0,3\} 
										 	  \textnormal{ and } N(v) = \{u_0, u_3\},\\
c_G(\{y_i,x_0\}) &=& i \in \{0,2\}, &\textnormal{and}\\
c_G(\{y_i,x_1\}) &=& i \in \{1,3\}.
\end{array}
\label{eqnColourGTwoTwoZeroZero}
\end{equation}

Now we show that $c_G$ is a rainbow colouring of $G$ by listing down a rainbow path between every pair of vertices which are at a distance of at least $2$ from each other. This time, let $I_i$, $i \in \Z_4$, denote the set of vertices in $I \setminus \{y_2,y_3\}$ with at least one neighbour in $K_i$. Notice that, as in the previous cases, the edge(s) between every $K_i$ and $K_j$, $i \neq j$, is given a colour different from $i$ and $j$. This ensures rainbow paths between $I_i$ and $K_j \cup I_j$, and we need to work hard only to identify rainbow paths from $y_2$ and $y_3$ to rest of the graph.

\begin{align*}
% \begin{array}{rcl}
\Path{I_i}{K_j \then I_j, \; i \neq j}
	&: I_i \Edge{i} K_i \Edge{k} K_j \Edge{j} I_j \; 
	(\textnormal{where }k \in \Z_4 \setminus \{i,j\})\\
\Path{I_{3,3}}{K_3 \then I_{3,3}}
	&: I_{3,3} \Edge{3} K_3 \Edge{1} K_3 \Edge{0} I_{3,3} \\
\Path{y_2}{I_0}
	&: y_2 \Edge{2} K_0 \Edge{0} I_0  \\
\Path{y_2}{K_2 \then K_1 \then I_1}
	&: y_2 \Edge{2} K_0 \Edge{3} K_2 \Edge{0} K_1 \Edge{1} I_1  \\
\Path{y_2}{K_3 \then I_3}
	&: y_2 \Edge{2} K_0 \Edge{1} K_3 \Edge{3} I_3  \\
\Path{y_2}{y_3}
	&: y_2 \Edge{2} K_0 \Edge{1} K_3 \Edge{0} K_1 \Edge{3} y_3  \\
\Path{y_3}{I_1}
	&: y_3 \Edge{3} K_1 \Edge{1} I_1  \\
\Path{y_3}{K_0 \then I_0}
	&: y_3 \Edge{3} K_1 \Edge{2} K_0 \Edge{0} I_0  \\
\Path{y_3}{K_2 \then I_2}
	&: y_3 \Edge{3} K_1 \Edge{0} K_2 \Edge{2} I_2  \\
\Path{y_3}{K_3 \then I_{3,3}}
	&: y_3 \Edge{3} K_1 \Edge{2} K_0 \Edge{1} K_3 \Edge{0} I_{3,3} 
% \end{array}
\end{align*}

Though we haven't indicated rainbow paths from $y_2$ to $I_2$ and $y_3$ to $I_3$, since $I_2 \subset I_0 \cup I_1 \cup I_3$ and $I_3 \subset I_0 \cup I_1 \cup I_2 \cup I_{3,3}$, we have exhausted all pairs of vertices in the list above.
\end{case}
\end{proof}

\begin{figure}
\begin{center}
\begin{pspicture}(0,0)(10,4)
	\rput{0}(2,2){\GTwoTwoZero}
	\uput[d](2,0){$G_{220}$}
	\rput{0}(8,2){\GTwoTwoZeroPrime}
	\uput[d](8,0){$G_{220}^z$}
\end{pspicture}
\end{center}
\caption{The graphs $G_{220}$ and $G_{220}^z$ mentioned in the proof of Corollary \ref{corSplitPoly}.}
\label{figGTwoTwoZero}
\end{figure}
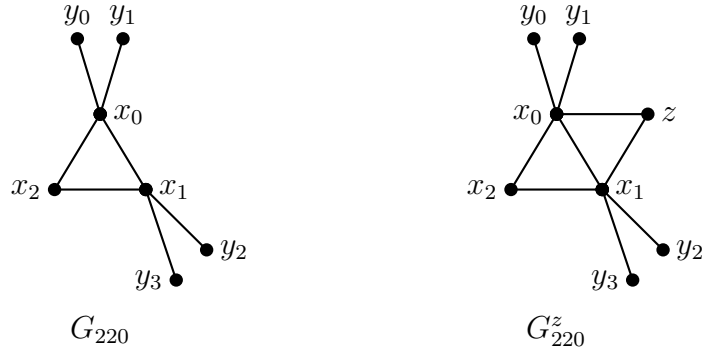

\begin{figure}
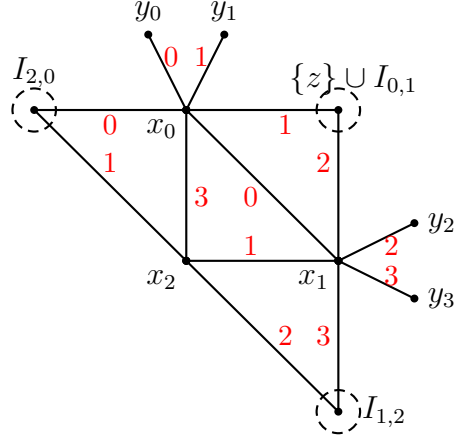

\begin{center}
\FigureColourGTwoTwoZeroPrime
\end{center}
\caption{A partial illustration of the colouring of $G$ when it contains $G_{220}^z$ as a subgraph. Part of proof of Corollary \ref{corSplitPoly}}
\label{figGTwoTwoZeroPrime}
\end{figure}

\begin{corollary} \label{corSplitPoly}
For each integer $k \geq 4$ there exists a polynomial time algorithm to check if the rainbow connection number of a split graph is at most $k$. Furthermore, any split graph with rainbow connection number at least $4$ can be optimally rainbow coloured in polynomial time.
\end{corollary}
\begin{proof}
Let $G$ be a split graph with $p$ pendant vertices. If $G$ is a tree, then $rc(G)$ is equal to the number of edges in $G$ and so we can check in linear time if $rc(G) \leq k$ for any $k$. In the case when $G$ is not a tree, a maximal clique $K$ in $G$ contains at least $3$ vertices. Fix any $k \geq 4$. If $p \leq k-1$, we know from \cite[Corollary $2$]{chandran2012split} that $rc(G) \leq k$. Similarly if $p > k$, then $rc(G) > k$. Hence we can assume that $p = k$. Let $K'$ be the vertices in $K$ which are adjacent to at least one pendant vertex of $G$. If $|K'| \geq 3$, then $G$ contains  $G_{111}$ as a subgraph with $\pen(G_{111}) \subset \pen(G)$. If  $|K'| \leq 2$ and $G[K \cup \pen(G)]$ is not isomorphic to $G_{220}$ (Figure \ref{figGTwoTwoZero}), then $G$ contains $H \in \{G_{400}, G_{310}, G_{2200}\}$ as a subgraph with $\pen(H) \subseteq \pen(G)$.  In all the cases above, it follows from Lemma \ref{lemmaSplitPendant} that  $rc(G) \leq k$ and the proof therein gives a rainbow colouring in polynomial time. 

If $G[K \cup \pen(G)]$ is isomorphic to $G_{220}$ then let us relabel $V(G)$ so that $G_{220}$ is a subgraph of $G$. It is not difficult to see that if $G$ has $G_{220}^z$ as a subgraph for some $z \in V(G)$ then the $rc(G) = 4$. See Figure \ref{figGTwoTwoZeroPrime} for a partial illustration of one possible rainbow colouring. Conversely, in any attempted rainbow colouring of $G$ using $4$ colours, the $4$ pendant edges $x_iy_j, i \in \{0,1\}, j \in \{0, \ldots, 3\}$ have to get $4$ different colours and the edge $x_0x_1$ has to reuse one of these $4$ colours, say the one used by $x_0y_0$ (as it is in Figure \ref{figGTwoTwoZeroPrime}). Then it is easy to see that we need at least two more $2$-length paths between $x_0$ and $x_1$ so as to provide rainbow paths from $y_0$ to $y_2$ and $y_3$, which is available only if $G$ has a subgraph isomorphic to $G_{220}^z$.  
\end{proof}

\subsection{Two colours: \NP-completeness}
\label{secTwoColours}

In order to show that $\rc(G,2)$ is \NP-complete, we will also use the following two decision problems:

\begin{problem}[$\bcc(G, \S)$]
Given a graph $G = (V,E)$, and a collection of subsets $\S$ of $V$, decide whether there exists a {\em bipartitioning function} $X : \S \into 2^V$ such that $X(T) \subset T, \forall T \in \S$ and for every edge $\{u,v\} \in E(G)$ there exists a $T \in \S$ with $u \in X(T)$ and $v \in T \setminus X(T)$. 
\end{problem}

\begin{problem}[$\sat(\varphi)$]
\label{problemSat}
Given a boolean formula $\varphi$ in which every clause contains exactly $3$ distinct literals corresponding to three distinct variables, decide whether there exists an evaluation of variables of $\varphi$ such that every clause contains at least one satisfied literal.
\end{problem}

Next two lemmata show a reduction of $\sat(\varphi)$ to $\bcc(G, \S)$ (Lemma \ref{lemma3SATtoBCC}) and a reduction of $\bcc(G, \S)$ to $\rc(G, 2)$ on split graphs (Lemma \ref{lemmaBCCtoRC}). Problem \ref{problemSat} is known to be \NP-complete since general \textsc{Sat} can be easily reduced to our version of $\sat$ problem. Note that all the three problems clearly belong to class \NP. It means that polynomial time reductions from $\sat(\varphi)$ to $\bcc(G, \S)$ and $\rc(G,2)$ is enough to show \NP-completeness of the latter. It can be easily seen that the reductions used in proofs of Lemmata \ref{lemmaBCCtoRC} and \ref{lemma3SATtoBCC} are polynomial. Thus we show that $\rc(G,2)$ is \NP-complete on split graphs (Theorem \ref{thm:NP}).

\begin{lemma}
\label{lemmaBCCtoRC}
Problem $\bcc(G, \S)$ is reducible to $\rc(G', 2)$ where $G'$ is a split graph.
\end{lemma}

\begin{proof}
%\Obr{0.25}{reduction_small.jpeg}{Graph $G' = (A', B', E')$.}{fig:reduction_old}
Let $(G,\S)$ be an instance of $\bcc$ and $G = (V,E)$.
Let $\bar E = \binom V 2 \setminus E$ be the set of edges of complement of~$G$.
We define split graph $G' = (A' \dot \cup B', E')$ in the following way (see Figure \ref{fig:reduction}):
$$A' = \bigcup_{v \in V}\{u'_v\} ~\cup~ \bigcup_{T \in \S} \{s_T\} ~\cup~ \bigcup_{e \in \bar E} \{x_e\}$$
$$B' = \bigcup_{v \in V}\{u_v\}$$
$$E' = \binom {A'} 2 \cup \bigcup_{v \in V} \{u_vu'_v\} \cup \bigcup_{v \in T \in \S} \{u_vs_T\} \cup \bigcup_{e \in \bar E;~ e = vw} \{u_vx_e, u_wx_e\}$$ 
\Obr{0.15}{biclique_to_rainbow.eps}{Graph $G' = (A' \dot \cup B', E')$.}{fig:reduction}

We prove that $rc(G') \leq 2$ if and only if $(G,\S)$ is a ``yes'' instance of $\bcc(G, \S)$. At first suppose that $(G,\S)$ is a ``yes'' instance of $\bcc$.
Let $X \colon \S \to 2^V$ be a function 
such that bi-cliques $\{(X(T), T \setminus X(T)):~ T \in \S\}$ cover all edges of $G$.
We define coloring $col \colon E' \to \{red, blue\}$ of edges of $G'$ in the following way:

\begin{itemize}
 \item $col(e') = blue$, if $e' \in \binom {A'} 2$
 \item $col(e') = red$, if $e' = u_vu'_v$ and $v \in V$
 \item $col(e') = blue$, if $e' = u_vs_T$, $T \in \S$ and $v \in X(T)$
 \item $col(e') = red$, if $e' = u_vs_T$, $T \in \S$ and $v \in T \setminus X(T)$
 \item For every $e = vw \in \bar E$ we set $col(u_vx_e)$ and $col(u_wx_e)$ in such a way that $col(u_vx_e) \neq col(u_wx_e)$.
 \end{itemize}

We will show that for this coloring there exists a rainbow path between any two vertices of $G'$. 
If $u,u' \in A'$ then $uu' \in E'$ and we are done. 
If $u_v \in B'$ and $u' \in A'$ then either $u' = u'_v$ and $u_vu' \in E'$ or we can take rainbow path $u_v, u'_v, u'$. 
If $u_v, u_w \in B'$ then either $e = vw \in \bar E$ and the path $u_v, x_e, u_w$ is rainbow or $vw \in E$, in which case there exists $T \in \S$ such that bi-clique $(X(T), T \setminus X(T))$ covers $vw$ and the path $u_v, s_T, u_w$ is rainbow.

\smallskip

For the opposite direction suppose that $col \colon E' \to \{red, blue\}$ is a coloring of edges of $G'$ such that there exists a rainbow path between any two vertices of $G'$.
We define mapping $X \colon \S \to 2^V$ in such a way that
$v \in X(T)$ if and only if $v \in T \in \S$ and $col(u_v, s_T) = blue$.
  
Suppose that $vw$ is an edge of $G$. From the definition of $G'$ we know that all paths from $u_v$ to $u_w$ in $G$ of length at most $2$ are of the form 
$u_v, s_T, u_{w}$ for some $T \in \S$. At least one of these paths has to be rainbow, say $u_v, s_T, u_{w}$. Then by the definition of $X(T)$ we know 
that bi-clique $(X(T), T \setminus X(T))$
covers edge $vw$ in $G$, what concludes the proof.
\end{proof}

\begin{lemma}
\label{lemma3SATtoBCC}
 Problem $\sat(\varphi)$ is reducible to $\bcc(G, \S)$.
\end{lemma}

\begin{proof}

Let $\varphi$ be an instance of $\sat$.
Let $v_1, \ldots, v_n$, resp. $C_1, \ldots, C_m$ be variables, resp. clauses of $\varphi$. 
Let $g \colon \{1, \ldots, m\} \times \{1,2,3\} \to \{1, \ldots, n\}$ be a function 
such that $v_{g(j,k)}$ is the variable corresponding to the $k$-th literal in clause $C_j$.
If variable $v$ has a positive, resp. negative appearance in clause $C$ then we write $v \in C$, resp. $\lnot v \in C$.
We will construct graph $G = (V,E)$ and family $\S \subseteq 2^V$  
such that $\varphi$ is satisfiable if and only if $(G,\S)$ is a ``yes'' instance of $\bcc$.

We start the construction of $G$ by defining $2$ types of vertices and $2$ types of edges (see Figure \ref{fig:E_12}):
$$ V_v = \bigcup_{i = 1, \ldots, n} \{a_i, f_i, f_i^1, f_i^2, t_i, t_i^1, t_i^2\}$$
$$ V_c = \bigcup_{j = 1, \ldots, m} \{A_j, F_j\}$$

$$ E_1 = \bigcup_{i = 1, \ldots, n} \{a_if_i, a_it_i, f_it_i^1, f_it_i^2, t_if_i^1, t_if_i^2, f_i^1t_i^1, f_i^2t_i^2\} $$
$$ E_2 = \bigcup_{j = 1, \ldots, m} \{A_jF_j\} \cup \bigcup_{j = 1, \ldots, m \atop k = 1,2,3} \{A_ja_{g(j,k)}, A_jt_{g(j,k)}, F_jt_{g(j,k)}\}$$

%\Obr{0.25}{E_1_and_E_2_small.jpeg}{Edges in sets $E_1$ and $E_2$.}{fig:E_12s}
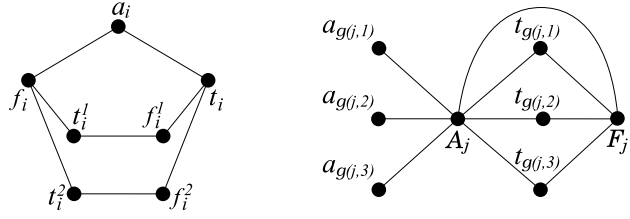
\Obr{0.4}{E_1_E_2.eps}{Edges in sets $E_1$ and $E_2$.}{fig:E_12}

%\Obr{0.25}{E_3_small.jpeg}{Edges in set $E_3$. Dashed edges are from $E_2$.}{fig:E_3}

Note that vertices of $V_v$ correspond to variables while the vertices of $V_c$  correspond to clauses of $\varphi$. 
We define $V = V_v \cup V_c$, $E = E_1 \cup E_2$ and $G = (V,E)$.

Next we define elements of the family $\S \subseteq 2^V$. 
For every $i = 1, \ldots, n$ we define (see Figures \ref{fig:V} and \ref{fig:C}):
$$ V_i^1 = \{a_i, f_i, f_i^1, t_i, t_i^1\} ~ \cup \bigcup_{v_i \in C_j} \{A_j\}$$
$$ V_i^2 = \{a_i, f_i, f_i^2, t_i, t_i^2\} ~ \cup \bigcup_{\lnot v_i \in C_j} \{A_j\}$$

%\Obr{0.25}{sets_V(i,j)_and_C(i,j)_small.jpeg}{Sets $V_i^l$ and $C_j^k$. Depicted edges are the edges that cannot be covered in any other set $S \in \S$.}{fig:VCs}
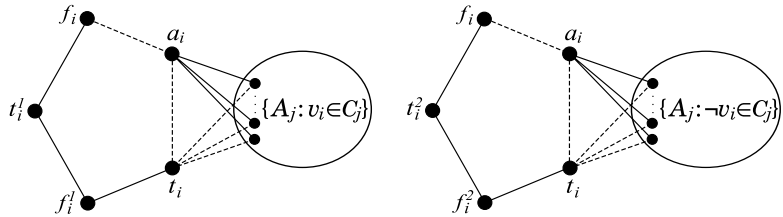
\Obr{0.08}{V_i.eps}{Sets $V_i^1$ and $V_i^2$. By dashed lines, resp. plain lines are depicted edges, resp. uniquely coverable edges of $G$.}{fig:V}

and for every $j = 1, \ldots, m$ and $k = 1,2,3$ we define:
$$ C_j^k = \{t_{g(j,k)}, A_j, F_j\}$$

\Obr{0.13}{C_jk.eps}{Set $C_j^k$. By dashed lines, resp. plain lines are depicted edges, resp. uniquely coverable edges of $G$.}{fig:C}

We conclude the construction of $\S$ by taking 
$$\S = \bigcup_{i = 1,\ldots, n} \{V_i^1, V_i^2\} ~ \cup ~ \bigcup_{j=1,\dots,m} \{C_j^1,C_j^2,C_j^3\}$$

Note that some edges of $G$ can be covered by only one bi-clique, since there is only one $T \in \S$ containing both endpoints of the given edge.
We say that those edges are \emph{uniquely coverable} (see Figures \ref{fig:V} and \ref{fig:C}).

At first suppose that $(G,\S)$ is a ``yes'' instance for  $\bcc$. 
Let $X \colon \S \to 2^V$ be the corresponding partitioning function.
Define function $Y \colon \S \to 2^V$ such that $Y(T) = T \setminus X(T)$ for every $T \in \S$.
Without loss of generality suppose that $a_i \in X(V_i^l)$ and $A_j \in X(C_j^k)$ for every feasible indices $i,j,k$ and $l$.
Define an evaluation $eval \colon \{v_1, \ldots, v_n\} \to \{false, ~true\}$ such that $eval(v_i) = true$ if and only if $t_i \in X(V_i^1)$.
We will show that $eval$ satisfies formula $\varphi$. First we prove the following four claims:

\begin{claim}[i] 
Let $i \in \{1,\ldots,n\}$ and $l \in \{1,2\}$. Then one of the following holds:
\begin{itemize}
\item $\{f_i, f_i^l\} \subseteq X(V_i^l)$ ~and~ $\{t_i, t_i^l\} \subseteq Y(V_i^l)$
\item $\{t_i, t_i^l\} \subseteq X(V_i^l)$ ~and~ $\{f_i, f_i^l\} \subseteq Y(V_i^l)$
\end{itemize}
\end{claim}

\begin{claim}[ii] 
Let $i \in \{1,\ldots,n\}$, $l \in \{1,2\}$ and let $A_j \in V_i^l$. Then $A_j \in Y(V_i^l)$.
\end{claim}

\begin{claim}[iii] 
Let $j \in \{1,\ldots,m\}$ and $k \in \{1,2,3\}$. Then one of the following holds:
\begin{itemize}
\item $F_j \in X(C_j^k)$ ~and~ $t_{g(j,k)} \in Y(C_j^k)$
\item $t_{g(j,k)} \in X(C_j^k)$ ~and~ $F_j \in Y(C_j^k)$
\end{itemize}
\end{claim}

Proof of Claims (i), (ii) and (iii) follow directly from the fact that corresponding edges are uniquely coverable (see Figures \ref{fig:V} and \ref{fig:C}).

\begin{claim}[iv]
Let $i \in \{1,\ldots,n\}$. Then one of the following holds:\\
\centerline{$\bullet$~ $\{a_i, f_i\} \subseteq X(V_i^1)$, $t_i \in Y(V_i^1)$, $\{a_i, t_i\} \subseteq X(V_i^2)$ ~and~ $f_i \in Y(V_i^2)$}\\
\centerline{$\bullet$~ $\{a_i, t_i\} \subseteq X(V_i^1)$, $f_i \in Y(V_i^1)$, $\{a_i, f_i\} \subseteq X(V_i^2)$ ~and~ $t_i \in Y(V_i^2)$}
\end{claim}

Edges $a_if_i$ and $a_it_i \in E(G)$ can only be covered by bi-cliques on sets $V_i^1$ and $V_i^2$.
Note that from Claim (i) follows that it is not possible to cover both edges $a_if_i$ and $a_it_i$ by the same bi-clique.
We also know that $a_i \in X(V_i^1)$ and $a_i \in X(V_i^2)$. 
It means that either $\{a_i, f_i\} \subseteq X(V_i^1)$ or $\{a_i, t_i\} \subseteq X(V_i^1)$ and these two cases correspond to the two cases of our Claim (iv).

Now we will show that using $eval$ there exists at least one positively evaluated literal in every clause of formula $\varphi$. 
Let $C_j$ be a clause of $\varphi$. 
We know that edge $A_jF_j \in E(G)$ is covered by some bi-clique. 
This bi-clique has to be on vertex set $C_j^1$, $C_j^2$ or $C_j^3$ (since those are the only sets of $\S$ containing both $A_j$ and $F_j$).
Suppose that $A_j \in X(C_j^k)$ and $F_j \in Y(C_j^k)$. 
By Claim (iii) we know that edge $A_jt_{g(j,k)} \in E(G)$ is not covered by the bi-clique on $C_j^k$ and has to be covered by the bi-clique on $V_i^1$ or $V_i^2$.
We will show that the $k$-th literal of clause $C_j$ is satisfied using the evaluation $eval$.

Let $i = g(j,k)$. 
If $v_i \in C_j$ (positive appearance of variable $v_i$) then the edge $A_jt_i \in E(G)$ has to be covered 
by bi-clique $(X(V_i^1), Y(V_i^1))$.
That is only possible if $t_i \in X(V_i^1)$ (by Claim (ii) we know that $A_j \in Y(V_i^1)$). 
It means that by the~definition $eval(v_i) = true$ and $k$-th literal of $C_j$ is satisfied.

If $\lnot v_i \in C_j$ then the edge $A_jt_i$ has to be covered by bi-clique $(X(V_i^2), Y(V_i^2))$.
That is only possible if $t_i \in X(V_i^2)$ (by Claim (ii)).
From Claim (iv) we have $t_i \notin X(V_i^1)$ what means that $eval(v_i) = false$.
Using the fact that $v_i$ has a negative appearance in $C_j$ we know that $k$-th literal of $C_j$ is satisfied.

\smallskip

In the rest of the proof suppose that $eval \colon \{v_1, \ldots, v_n\} \to \{false, true\}$ is a satisfying evaluation of formula $\varphi$.
We will show that $(G,\S)$ is a ''yes`` instance of $\bcc$.

For every $i \in \{1,\ldots, n\}$ we define:
\begin{itemize}
 \item if $eval(v_i) = true$: $X(V_i^1) = \{a_i,t_i,t_i^1\}$ and $X(V_i^2) = \{a_i,f_i,f_i^2\}$
 \item if $eval(v_i) = false$: $X(V_i^1) = \{a_i,f_i,f_i^1\}$ and $X(V_i^2) = \{a_i,t_i,t_i^2\}$
\end{itemize}

%\Obr{0.25}{partition_of_V_and_C_small.jpeg}{An example of partitions of sets $V_i^1$ and $V_i^2$ in case of $eval(v_i) = true$ and set $C_j^k$ in case that $k$-th literal of $c_j$ is evaluated to $true$.}{fig:partitions}

and for every $j \in \{1, \dots, m\}$ and $k \in \{1,2,3\}$ we define:
\begin{itemize}
 \item if the $k$-th literal of $C_j$ is evaluated to $true$: $X(C_j^k) = \{A_j, t_{g(j,k)}\}$
 \item if the $k$-th literal of $C_j$ is evaluated to $false$: $X(C_j^k) = \{A_j, F_j\}$
\end{itemize}

Note that for every $T \in \S$ we define $Y(T) = T \setminus X(T)$.
We will show that all edges of $G$ are covered by some bi-clique $(X(T), Y(T))$. 
From the definition of $eval$ we know that edges of $E_1$ corresponding to variable $v_i$ are always covered by bi-cliques on $V_i^1$ and $V_i^2$.
By definition these bi-cliques also cover all edges $A_ja_{g(j,k)}$. 

Edges $t_{g(j,k)}F_j$ are covered by bi-cliques on $C_j^k$.
To conclude the proof we need to prove that also all  edges $A_jF_j$ and $A_jt_{g(j,k)}$ are covered by some bi-cliques.

From our assumption we know that every clause $C_j$ contains at least one positively evaluated literal. 
For this literal we have $X(C_j^k) = \{A_j, t_{g(j,k)}\}$. It means that the edge $A_jF_j$ is covered by bi-clique $(X(C_j^k), Y(C_j^k))$.

Edge $A_jt_{g(j,k)}$ is covered by bi-clique $(X(C_j^k), Y(C_j^k))$ whenever the $k$-th literal of $C_j$ is evaluated to false by $eval$.
Suppose that corresponding literal is evaluated to true.
Let $i = g(j,k)$. If $v_i \in C_j$, resp. $\lnot v_i \in C_j$  then $eval(v_i) = true$, resp. $eval(v_i) = false$ what implies that 
edge $A_jt_{g(j,k)}$ is covered by bi-clique $(X(V_i^1), Y(V_i^1))$, resp. $(X(V_i^2), Y(V_i^2))$.
\end{proof}

Next Theorem merges the results of Lemmata \ref{lemmaBCCtoRC} and \ref{lemma3SATtoBCC}.

\begin{theorem}
\label{thm:NP}
The problem $\rc(G,2)$ is \NP-complete even when $G$ is restricted to be a split graph.
\end{theorem}

% \begin{proof}
% NP-completeness of {\sc $2$-Rainbow Coloring} follows directly from Lemmas \ref{lemmaBCCtoRC} and \ref{lemma3SATtoBCC}.
% \end{proof}
% 

Continuing with the notations introduced in the reduction from $\bcc(G, \S)$ to $\rc(G', 2)$ (Lemma \ref{lemmaBCCtoRC}), consider the bipartite graph $H(A'' \dot \cup B', E'')$ defined as follows:

$$A'' = \bigcup_{T \in \S} \{s_T\} ~\cup~ \bigcup_{e \in \bar E} \{x_e\}$$
$$B' = \bigcup_{v \in V}\{u_v\}$$
$$E'' = \bigcup_{v \in T \in \S} \{u_vs_T\} \cup \bigcup_{e \in \bar E;~ e = vw} \{u_vx_e, u_wx_e\}$$ 

It is easy to see that the same proof can be modified to show that the following problem is \NP-complete.

\begin{problem}[$\textsc{BipartiteRainbow}(H)$]
\label{problemBipartiteRainbow}
Given a bipartite graph $H$ with parts $A$ and $B$, decide whether the edges of $H$ can be $2$-coloured so that there exists a rainbow path between any two vertices in part $B$.
\end{problem}

The above problem is equivalent to Problem \ref{problemMatrix} ($\textsc{EnsureDistinctRows}(C)$) where $n$ is the size of part $A$, $m$ is the size of part $B$ and $C$ corresponds to the missing edges of $H$ across the bipartition. To see that Problem \ref{problemBipartiteRainbow} is equivalent to Problem \ref{problemPacking} ($\textsc{OrthogonalPacking}$), Fix an ordering of $(a_1, \ldots, a_n)$ of the vertices in $A$ and associate with each vertex $v \in B$ an $n$-dimensional box $b(v)$ of sides $(s_1, \ldots, s_n)$ where $s_i = 1/2$ if $va_i \in E(H)$ and $s_i = 1$ otherwise. Now a $\{red,~blue\}$ rainbow colouring of edges $H$ can be interpreted as the location of the ``left-bottom'' corner of the boxes in a packing of them into the unit cube. In particular, a box $b(v)$ of size $(s_1, \ldots, s_n)$ occupies the space $[x_1, x_1 + s_1] \times \cdots \times [x_n, x_n + s_n]$, where $x_i = 1/2$ if $va_i \in E(H)$ with colour $red$ and $x_i = 0$ otherwise. 

%%%%%%%%%%%% Section on Two colours ends

\end{document}